\documentclass[12pt]{article}
\usepackage{geometry} 
\usepackage[utf8]{inputenc}
\usepackage[T1]{fontenc}
\usepackage[american]{babel}
\usepackage{lmodern}
\usepackage{graphicx}
\usepackage{amsmath}   
\usepackage{amssymb}   
\usepackage{amsfonts}
\usepackage{amsthm}
\usepackage{epstopdf}
\usepackage{mathtools}
\usepackage[pdftex]{color}
\usepackage{xcolor}
\usepackage{bm}
\usepackage{tikz}
\usepackage{url}
\usetikzlibrary{matrix}
\usepackage{multicol}
\usepackage{array}
\usepackage[parfill]{parskip}
\usepackage{fancyhdr}
\usepackage{lipsum}
\usepackage{longtable}
\usepackage{enumerate}
\usepackage{subcaption}
\usepackage{pinlabel}
\usepackage{tikz-cd}
\usepackage{setspace} 

\usepackage[pdftex]{color}
\usepackage{xcolor}

\usepackage[parfill]{parskip}
\makeatletter
\def\thm@space@setup{%
  \thm@preskip=\parskip \thm@postskip=0pt
}
\makeatother

\numberwithin{equation}{section}
\newtheorem{de}{Definition}[section]
\newtheorem{theo}[de]{Theorem}
\newtheorem{pro}[de]{Proposition}

\newtheorem{lemm}[de]{Lemma}

\theoremstyle{remark}

\theoremstyle{plain}

\newcommand{\mbP}{\mathbb{P}}
\newcommand{\mbC}{\mathbb{C}}
\newcommand{\mbR}{\mathbb{R}}

\newcommand{\cA}{\mathcal A}
\newcommand{\cB}{\mathcal B}
\newcommand{\cP}{\mathcal P}
\newcommand{\cQ}{\mathcal Q}

\newcommand{\zti}{\widetilde{z}}

\allowdisplaybreaks

\title{\textbf{A three-dimensional generalization of QRT maps}}
\author{Jaume Alonso \and Yuri B.\ Suris\and Kangning Wei }
\date{\small Institut für Mathematik, MA 7-1\\ Technische Universität Berlin, Str.\ des 17.\ Juni, 10623 Berlin, Germany\\
E-mail: alonso@math.tu-berlin.de, suris@math.tu-berlin.de, wei@math.tu-berlin.de}

\begin{document}

\maketitle

\begin{abstract}
\noindent We propose a geometric construction of three-dimensional birational maps that preserve two pencils of quadrics. The maps act as compositions of involutions, which, in turn, act along the straight line generators of the quadrics of the first pencil and are defined by the intersections with quadrics of the second pencil. On each quadric of the first pencil, the maps act as two-dimensional QRT maps.\\

\noindent While these maps are of a pretty high degree in general, we find geometric conditions which guarantee that the degree is reduced to 3. The resulting degree 3 maps are illustrated by two known and two novel Kahan-type discretizations of three-dimensional Nambu systems, including the Euler top and the Zhukovski-Volterra gyrostat with two non-vanishing components of the gyrostatic momentum.
\end{abstract}

\section{Introduction}

The theory of integrable systems has vitally important interrelations with algebraic geometry and other branches of geometry, like projective and differential geometry. Some celebrated examples of discrete integrable systems are the so-called QRT maps, introduced in 1988 by Quispel, Roberts and Thompson \cite{QRT1988, QRT1989} and which enjoy a rich geometric structure. Originally formulated in terms of pencils of biquadratic curves, their relation to rational elliptic surfaces has been clarified in \cite{Veselov1991}, \cite{Tsuda2004}, and a monographic exposition of their multiple interrelations with algebraic geometry can be found in \cite{Duistermaat2010}.

Further examples of integrable planar maps with integrals of higher degrees were introduced in \cite{HKY_2001}, \cite{Kimura_et_al_2002} (sometimes they are quoted as HKY maps, after Hirota, Kimura, and Yahagi). Additional examples and constructions can be found in \cite{Tsuda_et_al_2007}, \cite{Tsuda_et_al_2009}, \cite{Kassotakis_Joshi_2010}. A certain classification of integrable planar maps related to rational elliptic surfaces was given 
in \cite{Carstea_Takenawa_2012}. It can be seen as a refinement of Sakai's classification of discrete Painlev\'e equations \cite{Sakai2001}. 

An extensive source of integrable birational maps in arbitrary dimension is provided by the so-called \emph{Kahan discretization method} of quadratic vector fields. This method, when applied to a system of ordinary differential equations with a quadratic vector field
\begin{equation}\label{defKH}
\dot{z}_k=  \sum_{i,j=1}^n a^{(k)}_{ij} z_i z_j + \sum_{i=1}^n b^{(k)}_i z_i + c^{(k)},\quad k=1,\ldots,n,
\end{equation}
results in a one-parameter family of birational maps $\widetilde z=f(z,\varepsilon)$ on $\mbC^n$, given by the following system of equations:
\begin{equation}\label{Kah1}
\dfrac{\widetilde{z}_k - z_k}{\varepsilon} = \dfrac{1}{2}\sum_{i,j=1}^n a^{(k)}_{ij}  ( \widetilde{z}_iz_j + z_i \widetilde{z}_j)+ \frac{1}{2}\sum_{i=1}^n b^{(k)}_i (\widetilde{z}_i + z_i) + c^{(k)}, \quad k=1,\ldots,n.
\end{equation} 
Equations \eqref{Kah1} are bilinear with respect to $(z_1,\ldots,z_n)$ and $(\widetilde z_1,\ldots,\widetilde z_n)$, so by Cramer's rule the map $f$ is rational of degree $n$ and the same is true for the map $f^{-1}$. Moreover, interchanging $z\leftrightarrow \widetilde z$ and $\varepsilon\leftrightarrow-\varepsilon$, we see that $f^{-1}(z,\varepsilon)=f(z,-\varepsilon)$. 

The Kahan discretization method was introduced in \cite{Kahan1993} as an unconventional numerical method with unexpectedly good stability properties. It reappeared in two seminal papers by Hirota and Kimura \cite{HirotaKimura2000, KimuraHirota2000}, who were apparently unaware of Kahan's work. In \cite{PPS2009, PPS2011}, it was established that this method tends to preserve integrability when applied to integrable systems. It was proposed to call this method ``Hirota-Kimura discretization'' in the context of integrable systems, but it seems that the name ``Kahan discretization'' remains more established. Remarkable integrability properties of this method were investigated in a number of further papers, including \cite{Celledoni1, Celledoni2, Celledoni3} and \cite{PS2019, PSS2019, PSWZ2021, SST2021}. 
 
The Kahan discretizations are better considered as birational maps $f:\mbP^n\dasharrow\mbP^n$ by regarding $(z_1,\ldots,z_n)$ as inhomogeneous coordinates on the affine part $\mbC^n\subset\mbP^n$, with $x=[z_1:\ldots:z_n:1]$. This type of maps produced by a set of $n$ bilinear relations between homogeneous coordinates of $x$ and $\widetilde x$ is well-known in the classical literature on birational (Cremona) maps. For instance, Cayley \cite{Cayley} calls such maps \emph{lineo-linear}, while Hudson \cite{Hudson1927} calls them \emph{bilinear}. In a more modern terminology \cite{Pan1997, Pan1999}, these are \emph{determinantal maps}, which means that the homogeneous coordinates of $f(x)$ are the $n+1$ minors of the maximal order $n$ of a $(n+1)\times n$ matrix whose entries are linear forms in $x$. In any case, this construction was classically used to produce birational maps of type $T_{n-n}$ (of bidegree $(n,n)$ in modern terms). 

In the classical literature on Cremona transformations, the idea of dynamics was alien, so that virtually no results on \emph{integrability} of birational maps of bilinear type can be found there. The study of integrability of Kahan discretizations can be considered as filling this gap, and constitutes one of our main motivations in the present paper. Our main guiding principle is the idea that geometry is underlying remarkable dynamical properties (compare with the monograph \cite{DDGbook} which puts discrete differential geometry to a basis of the theory of discrete integrable systems). Thus, we derive integrable maps via geometric constructions.

More precisely, we propose here a proper three-dimensional generalization of the class of QRT maps. It should be mentioned that several multidimensional generalizations of QRT map are available in the literature (see, e.g., \cite{Capel_Sahadevan_2001}, \cite{Tsuda2004}), but all of them try to reproduce the analytical mechanism of integrability on the level of formulas, which makes them miss --in our opinion-- the most interesting cases. Our idea is, first, to reinterpret QRT maps as maps on a quadric in $\mbP^3$ (which is done in Section \ref{sect QRT on one quadric}, after we recall the original QRT construction in $\mbP^1\times \mbP^1$ in Section \ref{sect QRT}), and then, second, to extend these to maps on a pencil of quadrics in $\mbP^3$, governed by the second pencil of quadrics (Section \ref{sec:construction}). From the point of view of integrable systems, the most appealing examples are the simplest ones, which here means the examples with the least degree. We find two geometric constructions which guarantee that the resulting maps are of degree 3, and actually of the bilinear (or determinantal) type. These constructions are given in Sections \ref{ss:dropdegree} and \ref{section: construction 2}, respectively. In the first construction, the both pencils of quadrics are of a special type (separable pencils). In the second construction, one of the pencils is separable and shares one common quadric with the second pencil, which can be arbitrary. Besides that, it is remarkable that our constructions include several previously known integrable systems and allow us to discover new ones. The previously known examples include the Kahan discretization of the Euler top (Section \ref{ss:dET}) and of the Zhukovski-Volterra gyrostat with one non-vanishing component of the gyrostatic momentum (Section \ref{sect: dZV beta1}). The new examples solve the problem of integrability of the Kahan discretization of a more general Zhukovski-Volterra gyrostat, which was open since it was first posed in \cite{PPS2011}, see Sections \ref{ss: dZV special} and \ref{ss: dZV general}.

\section{QRT maps on $\mbP^1\times\mbP^1$}
\label{sect QRT}

To quickly introduce QRT maps, consider a \emph{pencil of biquadratic curves}
$$
\cA_\lambda=\Big\{(x,y)\in \mbC^2: A_\lambda(x,y):=A_0(x,y)-\lambda A_\infty(x,y)=0\Big\},
$$
where $A_0,A_\infty$ are two polynomials of bidegree (2,2). The \emph{base set} $\mathcal B$ of the pencil is defined as the set of points through which all curves of the pencil pass or, equivalently, as the intersection $\{A_0(x,y)=0\}\cap\{A_\infty(x,y)=0\}$. 
Through any point $(x_0,y_0)\not\in\mathcal B$, there passes exactly one curve of the pencil, defined by $\lambda=\lambda(x_0,y_0)=A_0(x_0,y_0)/A_\infty(x_0,y_0)$. 

It is often better to consider this pencil in a compactification of $\mbC^2$, which could be chosen either as $\mbP^2$ or as $\mbP^1\times \mbP^1$. These two choices lead to a somewhat different geometry. 
\begin{itemize}
\item In $\mbP^1\times \mbP^1$, $\mathcal B$ consists of eight base points;
\item In $\mbP^2$, $\mathcal B$ consists of eight simple base points and two double base points $[1:0:0]$ and $[0:1:0]$ at infinity.
\end{itemize}
In this paper, we restrict ourselves to the case $\mbP^1\times \mbP^1$. 

One defines the \emph{horizontal switch} $i_1$ and the \emph{vertical switch} $i_2$ as follows. For a given point $(x_0,y_0)\in\mbP^1\times \mbP^1 \backslash \mathcal{B}$, determine $\lambda=\lambda(x_0,y_0)$ as above. Then the horizontal line $\{y=y_0\}$ intersects $A_\lambda$ at exactly one further point $(x_1,y_0)$ which is defined to be $i_1(x_0,y_0)$; similarly, the vertical line $\{x=x_0\}$ intersects $A_\lambda$ at exactly one further point $(x_0,y_1)$ which is defined to be $i_2(x_0,y_0)$. The QRT map is defined as
$$
g=i_2\circ i_1.
$$
Each of the maps $i_1$, $i_2$ is a birational involution on $\mbP^1\times \mbP^1$ with indeterminacy set $\mathcal B$. Likewise, the QRT map $f$ is a (dynamically nontrivial) birational map on $\mbP^1\times \mbP^1$, having $\lambda(x,y)=A_0(x,y)/A_\infty(x,y)$ as an integral of motion. On each fiber 
$$
\cA_{\lambda_0}=\Big\{(x,y): \lambda(x,y)=\lambda_0\Big\},
$$
which is generically an elliptic curve, $g$ acts as a shift with respect to the corresponding addition law.

In the important \emph{symmetric case}, where the polynomials $A_0$ and $A_\infty$ are symmetric under the reflection
$$
\sigma(x,y)=(y,x),
$$
one can define the so-called \emph{QRT root}, such that $g=f\circ f$, by the formula
$$
f=\sigma\circ i_1=i_2\circ\sigma.
$$

\section{QRT maps as maps on a quadric}
\label{sect QRT on one quadric}

We propose a map defined by the following geometric data:
\begin{itemize}
\item a non-degenerate quadric $\cP$ in $\mbP^3$,
\item and a pencil of quadrics $\cQ_\lambda$ in $\mbP^3$.
\end{itemize}
We now define two involutive maps $i_1,i_2$ on $\cP$ as follows. The quadric $\cP$ admits two rulings such that any two lines of one ruling are skew and any line of one ruling intersects any line of the second ruling. Through each point $X\in\cP$ there pass two straight lines, one of each of the two rulings, let us call them $\ell_1(X)$ and $\ell_2(X)$.  The quadric $\cP$ can be considered as a fibration by the quartic intersection curves $\cA_\lambda=\cP\cap\cQ_\lambda$. For a given point $X\in\cP$, let $\lambda=\lambda(X)$ be defined as the value of the pencil parameter for which $X\in \cA_\lambda$. Denote by $i_1(X)$, $i_2(X)$ the second intersection point of $\ell_1(X)$ with $\cQ_\lambda$, resp. the second intersection point of $\ell_2(X)$ with $\cQ_\lambda$.

It is easy to see that these maps are isomorphic to the corresponding QRT switches. Indeed, any non-degenerate quadric $\cP$ in $\mbP^3$ is written in suitable coordinates as 
$$
\cP=\Big\{[X_1:X_2:X_3:X_4]: X_1X_2=X_3X_4\Big\}\subset \mbP^3.
$$
It is isomorphic to $\mbP^1\times\mbP^1$, via
$$
\mbP^1\times\mbP^1\ni (x,y)=\big([x_1:x_0],[y_1:y_0]\big)\;\;\mapsto \;\;[x_0y_0:x_1y_1:x_1y_0:x_0y_1]\in \cP.
$$
In other words, in affine coordinates on $\mbP^1\times\mbP^1$:
$$
X_2:X_1=xy, \quad X_3:X_1=x, \quad X_4:X_1=y.
$$
Thus, the intersection of an arbitrary quadric $\cQ\in\mbP^3$ of equation
\begin{align*}
& a_{11}X_1^2+a_{22}X_2^2+a_{33}X_3^2+a_{44}X_4^2\\
& +a_{12}X_1X_2+a_{13}X_1X_3+a_{14}X_1X_4 
   +a_{23}X_2X_3+a_{24}X_2X_4+a_{34}X_3X_4 =0
\end{align*}
with $\cP$ corresponds to a biquadratic curve in $\mbP^1\times \mbP^1$ of equation
$$
a_{11}+(a_{12}+a_{34})xy+a_{13}x+a_{14}y 
+a_{22}x^2y^2+a_{23}x^2y+a_{24}xy^2 
+a_{33}x^2+a_{44}y^2 =0.
$$
Therefore, the fibration of $\cP$ by curves $\cA_\lambda=\cP\cap\cQ_\lambda$ corresponds to a fibration of $\mbP^1\times \mbP^1$ by a pencil of biquadratic curves.

On the other hand, the generators of $\cP$ through $X=[X_1:X_2:X_3:X_4]\in\cP$ are easily computed (see Lemma \ref{lemma generators}) and are given at a generic point $X$ by formulas \eqref{l1}, \eqref{l2} below. In affine coordinates $(x,y)$ on $\mbP^1\times \mbP^1$, these formulas turn into
$$
\ell_1(x,y)=\big\{(tx,y) : t\in\mbC\big\} \quad {\rm and} \quad \ell_2(x,y)=\big\{(x,ty): t\in\mbC\big\}
$$
respectively (the first formula holds true if $x\neq 0$, the second one if $y\neq 0$). These are horizontal, resp. vertical lines through $(x,y)$ Therefore, involutions $i_1, i_2$ on $\cP$ along generators of $\cP$ correspond to the horizontal, resp.  vertical switch in $\mbP^1\times\mbP^1$.

Summarizing, we come to the following definition.
\begin{de} \label{def QRT quadric}
For any nondegenerate quadric $\cP$ and a pencil of quadrics $\cQ_\lambda$, the  \emph{QRT map} $g:\cP\to\cP$ is defined as
$$
g=i_2\circ i_1.
$$
If the pencil $\cQ_\lambda$ is symmetric with respect to 
$$
\sigma(X_1,X_2,X_3,X_4)=(X_1,X_2,X_4,X_3),
$$
then we define the \emph{QRT root} $f:\cP\to \cP$, such that $g=f\circ f$, by the formula
$$
f=\sigma\circ i_1=i_2\circ\sigma.
$$
\end{de}

\section{3D generalization of the QRT construction}
\label{sec:construction}

We now generalise the QRT construction to the three-dimensional space.

\begin{de} 
Given two pencils of quadrics $\cP_\mu$ and $\cQ_\lambda$, we say that a birational map $g:\mbP^3\dasharrow \mbP^3$ is a 3D generalization of QRT if it leaves all quadrics of both pencils invariant, and induces on each $\cP_\mu$ a QRT map, according to Definition \ref{def QRT quadric}.
\end{de}
However, this definition is not that easy to realize. The main difficulty is to ensure that $g$ is birational. Indeed, for a generic pencil $\cP_\mu$, generators of the quadrics of the pencil are not rational functions in $\mbP^3$.

\textbf{Counterexample.} Let $\cP_\mu$ be the pencil
$$
\cP_\mu=\big\{ X_1^2+X_2^2+X_3^2-\mu X_4^2=0\big\}.
$$
Let us consider an affine space with inhomogeneous coordinates $(X_1,X_2,X_3)$ by requiring $X_4=1$, so that 
$$
X_1^2+X_2^2+X_3^2=\mu.
$$ 
We look for the straight line generators of the quadric $\cP_\mu$ through the point $X$, i.e., we look for vectors $(V_1,V_2,V_3)$ such that $(X_1 + t V_1, X_2 + tV_2 , X_3 + tV_3)$ belongs to the quadric for all $t$. This means, that
$$ (X_1 + tV_1)^2+(X_2 + tV_2)^2+(X_3 + tV_3)^2 =\mu,$$ which is a quadratic equation in $t$. Equating the coefficients of this equation to 0, we obtain
\begin{equation}
\begin{cases}
X_1 V_1 + X_2 V_2 + X_3 V_3 =0 ,\\
V_1^2+ V_2^2+V_3^2=0.
\end{cases}
\end{equation}
From this, we get directions of the two generators through $[X_1:X_2:X_3:1]$:
$$
[V_1:V_2:V_3]=\left[\frac{-X_1X_2\pm i\sqrt{\mu}X_3}{X_1^2+X_3^2}:1:\frac{-X_2X_3\mp i\sqrt{\mu}X_1}{X_1^2+X_3^2}\right].
$$
Thus, for any fixed $\mu$, we get the directions of the straight line generators as rational functions on $\mbP^3$, but for the pencil as a whole, these expressions depend on $\sqrt{X_1^2+X_2^2+X_3^2}$, i.e., are non-rational.

We will not give a complete characterization of pencils for which this dependence is rational, and restrict ourselves in this paper to one particularly interesting case.

\begin{de}
A pencil of quadrics $\cP_\mu$ in $\mbP^3$ is called \emph{separable} if it contains two reducible quadrics, each consisting of two planes, $\cP_0=\Pi_1\cup\Pi_2$ and $\cP_\infty=\Pi_3\cup\Pi_4$, all four planes being distinct.
\end{de}
Choosing those four planes as coordinate planes, we come to the following formula for a separable pencil:
\begin{equation}\label{sep2}
\cP_\mu=\Big\{ X_1X_2-\mu X_3X_4=0\Big\}.
\end{equation}
This pencil can be characterized by having the base set consisting of four lines
\begin{eqnarray}
L_1= \{X_1 = X_3 = 0\}, & \; & L_2= \{X_1 = X_4 =0\}, \label{lines 12}\\
L_3 =\{X_2 = X_4 =0\}, & \; & L_4= \{X_2 = X_3 = 0\}. \label{lines 34}
\end{eqnarray}
Through each point in $\mbP^3$ not belonging to the base set, there passes exactly one quadric of the pencil.   

We now compute the straight line generators of the quadrics $\cP_\mu$.
\begin{lemm}\label{lemma generators}
The straight line generators of $\cP_\mu$ through a generic point $X\in\cP_\mu$ are given by 
\begin{eqnarray}
\ell_1(X) & = & \big\{[X_1:tX_2:tX_3:X_4] : t\in\mbP^1\big\}, \label{l1} \\
\ell_2(X) & = & \big\{[X_1:tX_2:X_3:tX_4]: t\in\mbP^1\big\}. \label{l2}
\end{eqnarray}
Formula \eqref{l1} for $\ell_1(X)$ is well defined unless $X$ belongs to either of the lines $L_2=\{X_1 = X_4=0\}$ or $L_4=\{X_2 = X_3 =0\}$.  Likewise, formula \eqref{l2} for $\ell_2(X)$ is well defined unless $X$ belongs to either of the lines $L_1=\{X_1 = X_3=0\}$ or $L_3=\{X_2 = X_4 =0\}$. \end{lemm}
\begin{proof}
We work in an affine part of $\mbP^3$ given by $X_1=1$, and look for vectors $(V_2,V_3,V_4)$ such that $(1, X_2 + t V_2, X_3 + tV_3 , X_4 + tV_4)$ belongs to the quadric for all $t$. This means that
$$ 
(X_2 + tV_2) - \mu(X_3 + tV_3)(X_4+tV_4) =0,
$$ 
which is a quadratic equation in $t$. Equating coefficients of this equation to 0, we obtain
\begin{equation}
\begin{cases}
V_2-\mu(X_3 V_4 + X_4 V_3)=0 ,\\
V_3 V_4=0.
\end{cases}
\end{equation} This equation gives us two generators:
\begin{itemize}
	\item The first generator $\ell_1(X)$ is obtained by setting $V_4=0$, so that we can take 
$$
[V_2:V_3] = [\mu X_4:1]=[X_2:X_3]
$$ 
(the second expression only holds true provided that $(X_2,X_3)\neq (0,0)$). 	
	\item The second generator $\ell_2(X)$ is obtained by setting $V_3=0$, so that we can take 
$$
[V_2:V_4] = [\mu X_3:1]=[X_2:X_4]
$$ 
(the second expression only holds true provided that $(X_2,X_4)\neq (0,0)$). 
\end{itemize}	
\end{proof}

\textbf{Remark.} One easily sees that each of the lines $L_1$, $L_3$ is a generator of the family $\ell_1$, while each of the lines $L_2$, $L_4$ is a generator of the family $\ell_2$.

We are now in a position to compute QRT involutions defined by a separable pencil $\cP_\mu$ and the second pencil
\begin{equation}\label{pencil2}
\cQ_\lambda := \left\lbrace [X_1 : X_2 : X_3 : X_4] \in \mbP^3 : Q_0 - \lambda Q_\infty =0 \right\rbrace,
\end{equation} 
where $Q_0 = Q_0(X_1,X_2,X_3,X_4)$ and $Q_\infty = Q_\infty(X_1,X_2,X_3,X_4)$ are two linearly independent quadratic forms. 

\begin{pro}
Involutions $i_1$, $i_2$ along generators of the separable pencil \eqref{sep2} defined by the pencil \eqref{pencil2} are given by:
\begin{eqnarray}
i_1: [X_1 : X_2 : X_3 : X_4] &  \mapsto & [X_1T_2  : X_2T_0  : X_3T_0  : X_4T_2], \label{i1}\\
i_2: [X_1 : X_2 : X_3 : X_4]  & \mapsto & [X_1S_2  : X_2S_0  : X_3S_2  : X_4S_0], \label{i2}
\end{eqnarray}
where
\begin{eqnarray}
T_2 & = & Q_0(0,X_2,X_3,0)Q_\infty(X_1,X_2,X_3,X_4)  - Q_0(X_1,X_2,X_3,X_4)Q_\infty(0,X_2,X_3,0), \qquad\quad\label{defT2}\\
T_0 & = & Q_0(X_1,0,0,X_4)Q_\infty(X_1,X_2,X_3,X_4) - Q_0(X_1,X_2,X_3,X_4)Q_\infty(X_1,0,0,X_4), \label{defT0}
\end{eqnarray}
and
\begin{eqnarray}
S_2 & = & Q_0(0,X_2,0,X_4)Q_\infty(X_1,X_2,X_3,X_4) - Q_0(X_1,X_2,X_3,X_4)Q_\infty(0,X_2,0,X_4), \qquad\quad\label{defS2} \\
S_0 & = & Q_0(X_1,0,X_3,0)Q_\infty(X_1,X_2,X_3,X_4) - Q_0(X_1,X_2,X_3,X_4)Q_\infty(X_1,0,X_3,0). \label{defS0}
\end{eqnarray}
\end{pro}
\begin{proof}
Since the derivations are similar, we give details for $i_1$ only. Take a generic point $X$ for which the generator $\ell_1(X)$ is given by \eqref{l1}, that is, $X$ not belonging to the lines $\{X_1 = X_4=0\}$ or $\{X_2 = X_3 =0\}$. Suppose also that $X$ does not belong to the base set of the pencil $\cQ_\lambda$. Determine the unique value of $\lambda=\lambda(X)=Q_0(X)/Q_\infty(X)$ such that $X\in\cQ(\lambda)$. Assuming that the line $\ell_1(X)$ does not lie on $\cQ_\lambda$, it will intersect  $\cQ(\lambda)$ at exactly one further point $i_1(X)$. To compute it, we solve for $t$ the quadratic equation
$$ 
Q_0(X_1,tX_2,tX_3,X_4)-\lambda Q_\infty(X_1,tX_2,tX_3,X_4)=0,
$$
\normalsize
 or, by substituting $\lambda=\lambda(X)$,
\begin{align*}
&Q_0(X_1,tX_2,tX_3,X_4)Q_\infty(X_1,X_2,X_3,X_4)\\
&-Q_0(X_1,X_2,X_3,X_4)Q_\infty(X_1,tX_2,tX_3,X_4)=0.
\end{align*}
This is a quadratic equation for $t$ of the form
$$ 
T_2(X_1,X_2,X_3,X_4)t^2 + T_1(X_1,X_2,X_3,X_4)t + T_0(X_1,X_2,X_3,X_4)=0,
$$ 
where $T_i$, $i=0,1,2$ are homogeneous polynomials of degree 4, and $T_2$, $T_0$ are explicitly given by \eqref{defT2}, \eqref{defT0}.
One of the solutions of the quadratic equation is $t=1$, so the other one is $t = T_0/T_2$. This gives \eqref{i1}.
\end{proof}
 
\textbf{Remark.}  If the polynomials $Q_0$, $Q_\infty$ are symmetric with respect to $\sigma: X_3\leftrightarrow X_4$, then $i_2\circ\sigma=\sigma\circ i_1$ defines the corresponding QRT root.

\section{Drop of degree: two separable pencils}
\label{ss:dropdegree}

The involutions $i_1$, $i_2$ are in general of degree 5. However, we are interested in examples of a possibly low degree. In this section, we give geometric conditions under which these involutions are of degree 3. In our examples the symmetry $\sigma$ is also present, so that we are able to construct integrable 3D generalizations of the QRT root $f=i_2\circ \sigma=\sigma\circ i_1$ of degree 3. Thus, we restrict our attention to the involution $i_1$ given by \eqref{i1}, and our goal will be to specify conditions under which the quartic polynomials $T_2$, $T_0$ have  a common factor of degree 2:
\begin{equation}\label{T0 T2 fact}
 T_0 = AB_0, \qquad T_2 = AB_2,
 \end{equation}
 where all three polynomials $A$, $B_0$, and $B_2$ are of degree 2, so that the involution $i_1$ is of degree 3. 

We will use the following notation for quadrics in $\mbP^3$:
 $$
 \cA=\{A=0\}, \quad \cB_0=\{B_0=0\}, \quad \cB_2=\{B_2=0\}.
 $$
Condition \eqref{T0 T2 fact} means that the variety $\{T_2=0\}\cap\{T_0=0\}$ consists of the quadric $\cA$ and of the set $\cB_0\cap\cB_2$ which is a part of the indeterminacy set $I(i_1)$. Vanishing of all coefficients of the quadratic equation $T_2t^2+T_1t+T_0=0$ means geometrically that the line $\ell_1(X)$ belongs to the quadric $\cQ_\lambda$ with $\lambda=\lambda(X)$.  Since the distinct lines $\ell_1(X)$ for  $X\in \cA$ do not intersect each other, these lines form one ruling of the quadric $\cA$. We now want to find a direct geometric characterization of this ruling. 

First of all, we observe that, as it follows from \eqref{defT2}, \eqref{defT0}, both polynomials $T_2$ and $T_0$ vanish on the skew lines $L_2=\{X_1=X_4=0\}$ and $L_4=\{X_2=X_3=0\}$. Moreover, $L_2$ is a double line of $\{T_0=0\}$, while $L_4$ is a double line of $\{T_2=0\}$. Two common ways to satisfy these conditions are:
\begin{enumerate}
\item[(1)] $L_2$ is a simple line of $\cB_0$, $L_4$ is a simple line of $\cB_2$, and both are simple lines of $\cA$;
\item[(2)] $L_2$ is a double line of $\cB_0$ (so that $L_0$ belongs to $\cA$), and $L_4$ is a double line of $\cB_2$ (so that $L_2$ belongs to $\cA$). In this case $\cB_0,\cB_2$ must be degenerate, while $\cA$ is non-degenerate. 
\end{enumerate}
In both cases $L_2,L_4$ belong to $\cA$. Since each $\ell_1(X)$ intersects $L_2$ and $L_4$ transversally, the latter two lines must belong to the other ruling of the quadric $\cA$.

To proceed, we show how case (1) above can be realized. For this goal, we make an additional assumption that the second pencil of quadrics $\cQ_\lambda$ is separable as well. In other words,
\begin{equation}\label{sep3}
Q_0=U_1 U_2, \qquad Q_\infty= U_3 U_4,
\end{equation} 
where $U_i = U_i(X_1,X_2,X_3,X_4)$, $i=1,\ldots,4$ are four linearly independent linear forms. The base set of the pencil \eqref{sep3} consists of the four pairwise intersecting lines:
\begin{eqnarray}
L_5=\{U_1 = U_3 = 0\}, & \; & L_6=\{U_1 = U_4 =0\}, \label{lines 56}\\
L_7 =\{U_2 = U_4 = 0\}, & \; & L_8= \{U_2 = U_3 =0\}. \label{lines 78}
\end{eqnarray} 

Condition $\ell_1(X)\subset \cQ_{\lambda(X)}$ means $\ell_1(X)$ must also intersect either the pair of lines $L_5,L_7$ or $L_6,L_8$, since these are two pairs of the base lines of the pencil $\cQ_\lambda$ belonging to two complementary rulings of each $\cQ_{\lambda(X)}$. Thus, either the lines $L_6$, $L_8$ or the lines $L_5$, $L_7$ must also belong to the ruling of the quadric $\cA$ complementary to that consisting of $\ell_1(X)$. 

Summarizing, we see that under the condition \eqref{T0 T2 fact} and if the second pencil of quadrics is separable, either the four lines $L_2$, $L_4$, $L_6$, $L_8$ or the four lines $L_2$, $L_4$, $L_5$, $L_7$ must belong to one ruling of the quadric $\cA$. Thus, a necessary condition for a factorization as in \eqref{T0 T2 fact} is that one of these quadruples of pairwise skew lines lie on a common quadric. We now show that this necessary condition is sufficient as well.

\begin{theo} \label{geprop}
Let $U_i(X)$, $i=1,\ldots,4$, be linearly independent linear forms, and let $\cQ_\lambda=\{U_1U_2-\lambda U_3U_4=0\}$ be the corresponding separable pencil of quadrics. Define the lines $L_1,\ldots, L_4$ as in \eqref{lines 12}, \eqref{lines 34}, and the lines $L_5,\ldots,L_8$ as in  \eqref{lines 56}, \eqref{lines 78}. Then both quartic polynomials $T_0,T_2$ defined in \eqref{defT0}, \eqref{defT2} are divisible by $A$, as in \eqref{T0 T2 fact}, if and only if one of the following two conditions is satisfied:
\begin{itemize}

\item[{\rm(a)}] either the four lines $L_2,L_4,L_6,L_8$ are pairwise skew and lie on a common quadric $\{A=0\}$. In this case $\{B_0=0\}$ passes through $L_2, L_5, L_7$, and $\{B_2=0\}$ passes through $L_4,L_5,L_7$;

\item[{\rm(b)}] or the four  lines $L_2,L_4,L_5,L_7$ are pairwise skew and lie on a common quadric $\{A=0\}$. In this case $\{B_0=0\}$ passes through $L_2, L_6, L_8$, and $\{B_2=0\}$ passes through $L_4,L_6,L_8$.

\end{itemize}
In both cases, the involution $i_1$ along the generators of the pencil $\cP_\mu=\{X_1X_2-\mu X_3X_4=0\}$ defined by the intersections with the pencil $\cQ_\lambda=\{U_1U_2-\lambda U_3U_4=0\}$ is given by   
\begin{equation}  \label{i1 red}
i_1: [X_1 : X_2 : X_3 : X_4]  \mapsto [X_1B_2  : X_2B_0  : X_3B_0  : X_4B_2 ],
\end{equation}
and has degree 3.
\end{theo}
\begin{proof}
Necessity is already shown. To prove the converse statement, assume, for the sake of definiteness, that the four pairwise skew lines $L_2,L_4,L_6,L_8$ lie on a common quadric $\cA=\{A=0\}$ (case (a)).  We have to show that both $T_2$, $T_0$ are divisible by $A$. For this, we have to show that for any $X\in\cA$, the line $\ell_1(X)$ lies on $\cQ_{\lambda(X)}$. Consider the ruling of the quadric $\cA$ complementary to the one containing $L_2,L_4,L_6,L_8$. The line of this ruling through a given point $X\in\cA$ can be alternatively defined either as the unique line through $X$ intersecting $L_2, L_4$, that is, the line $\ell_1(X)$, or as the unique line through $X$ intersecting $L_6, L_8$, that is, the corresponding generator of $\cQ_{\lambda(X)}$. This proves the statement above. To demonstrate the statements about $\cB_0$, $\cB_2$, we observe that, by definition, both $T_0$ and $T_2$ vanish along all four base lines $L_5,L_6,L_7,L_8$, as well as along $L_2$ and $L_4$, where (as pointed out at the beginning of the present section) $L_2$ is a double line of $\{T_0=0\}$, while $L_4$ is a double line of $\{T_2=0\}$.
\end{proof}

\textbf{Remark.} Cases (a) and (b) of Theorem \ref{geprop} are brought one into each other by a simple renaming $U_3\leftrightarrow U_4$, which leads to $L_5\leftrightarrow L_6$, $L_7\leftrightarrow L_8$, but does not change the geometry.

\begin{pro} \label{prop indeterminacy set}
Assume that the pencils $\cP_\mu$ and $\cQ_\lambda$ are in general position, i.e., each of the base lines $L_5,\ldots,L_8$ of the second pencil is pairwise skew to each of the lines $L_2$, $L_4$. Then:
\begin{itemize}
\item in the case {\rm(a)}, the indeterminacy set of the involution $i_1$ is 
$$
 I(i_1) = L_2 \cup L_4 \cup L_5 \cup L_7 \cup L_9 \cup L_{10},
 $$ 
 where $L_9,L_{10}$ are the two lines intersecting all four skew lines $L_2,L_4,L_5,L_7$;
 \item in the case {\rm(b)}, the indeterminacy set of the involution $i_1$ is 
$$
 I(i_1) = L_2 \cup L_4 \cup L_6 \cup L_8 \cup L_9 \cup L_{10},
 $$ 
 where $L_9,L_{10}$ are the two lines intersecting all four skew lines $L_2,L_4,L_6,L_8$.
 \end{itemize}
\end{pro}
\begin{proof}
We consider, for definiteness, the case (a). 

Given four pairwise skew lines $L_2,L_4,L_5,L_7$, there exist exactly two lines, say $L_9$ and $L_{10}$, that intersect all four. To show this, recall that we can define $\cB_2$ as the quadric through the three pairwise skew lines $L_4,L_5,L_7$. These three lines belong to one ruling of $\cB_2$. The line $L_2$ intersects $\cB_2$ at two points. Let now $L_9$ and $L_{10}$ be the lines from the second ruling of $\cB_2$ through those two points. Then they intersect all four lines $L_2,L_4,L_5,L_7$. And, by construction, they lie on $\cB_2$.

Now consider the quadric $\cB_0$ through $L_2,L_5,L_7$. Exactly as above, we show that $L_9,L_{10}$ lie on $\cB_0$. Therefore, all four lines $L_5,L_7,L_9,L_{10}$ belong to the intersection $\cB_0 \cap \cB_2$ and, since this intersection is a curve of degree 4, it coincides with those four lines. Thus, these four lines belong to the indeterminacy set $I(i_1)$. 

Finally, we see from \eqref{i1 red} that the line $L_2=\{X_1=X_4=0\}$, where $B_0=0$, and the line $L_4=\{X_2=X_3=0\}$, where $B_2=0$, also belong to $I(i_1)$. And these six lines exhaust $I(i_1)$, since the indeterminacy set of a birational 3-dimensional map of degree 3 is a curve of degree 6.
\end{proof}

\section{Drop of degree: two pencils with one common quadric}
\label{section: construction 2}

Here we consider another possibility to achieve that the involutions $i_1$, $i_2$ be of degree 3, realizing case (2) mentioned at the beginning of Section \ref{ss:dropdegree}. Suppose that the pencils $\cP_\mu$ and $\cQ_\lambda$ have one common quadric, say $\cP_{\mu_0}$. 
\begin{theo} \label{Theorem construction 2}
If
\begin{equation}
Q_0=X_1X_2-\mu_0X_3X_4,
\end{equation}
and $Q_\infty(X_1,X_2,X_3,X_4)$ is an arbitrary homogeneous polynomial of degree 2, then the polynomials $T_2$, $T_0$ admit a factorization as in \eqref{T0 T2 fact}, with
\begin{equation}
A=-Q_0(X_1,X_2,X_3,X_4),
\end{equation}
and 
\begin{eqnarray}
B_2 & = & Q_\infty(0,X_2,X_3,0),\\
B_0 & = & Q_\infty(X_1,0,0,X_4),
\end{eqnarray}
so that each of the quadrics $\cB_2=\{B_2=0\}$ and $\cB_0=\{B_0=0\}$ is a pair of planes.
In this case, the involution $i_1$ along the generators of the pencil $\cP_\mu=\{X_1X_2-\mu X_3X_4=0\}$ defined by the intersections with the pencil $\cQ_\lambda$ is given by    \eqref{i1 red} and has degree 3. The indeterminacy set $I(i_1)$ consists of the four lines $\cB_0\cap\cB_2$, and of the two lines $L_2=\{X_1=X_4=0\}$ and $L_4=\{X_2=X_3=0\}$. These six lines form the side lines of a tetrahedron.
\end{theo}
\begin{proof}
This follows directly from \eqref{defT2}, \eqref{defT0}, upon taking into account that 
$$
Q_0(0,X_2,X_3,0)=0,\quad Q_0(X_1,0,0,X_4)=0.
$$
The statement about $I(i_1)$ follows immediately, after observing that $B_0$ and $B_2$ depend only on two variables each and therefore are factorisable into linear factors.
\end{proof}

\section{Example I: Kahan discretization of the Euler top}
\label{ss:dET}

The \emph{Euler top} (ET) is a free rigid body rotating around a fixed point. The evolution of the components of the angular momentum of ET in the moving frame is described by the following system:
\begin{equation}\label{ET}
\begin{cases} \dot{z}_1 = \alpha_1 z_2 z_3, \\ \dot{z}_2 = \alpha_2 z_3 z_1, \\ \dot{z}_3 = \alpha_3 z_1 z_2. \end{cases}
\end{equation}
This is an integrable system with the following constants of motion:
\begin{equation}
H_1 = \alpha_2 z_3^2 - \alpha_3 z_2^2,\qquad
H_2 = \alpha_3 z_1^2 - \alpha_1 z_3^2,\qquad
H_3 = \alpha_1 z_2^2 - \alpha_2 z_1^2.
\end{equation} 
Only two of them are functionally independent, since $\alpha_1 H_1 + \alpha_2 H_2 + \alpha_3 H_3 =0$. 
Its Kahan-style discretization was first introduced by Hirota and Kimura \cite{HirotaKimura2000}, and is defined by the following implicit equations of motion:
\begin{equation}\label{dET implicit} 
\begin{cases} \widetilde{z}_1-z_1 = \varepsilon \alpha_1 (\widetilde{z}_2 z_3 + z_2 \widetilde{z}_3), \\ 
\widetilde{z}_2-z_2 = \varepsilon \alpha_2 (\widetilde{z}_3 z_1 + z_3 \widetilde{z}_1), \\ 
\widetilde{z}_3-z_3 = \varepsilon \alpha_3 (\widetilde{z}_1 z_2 + z_1 \widetilde{z}_2).
\end{cases}
\end{equation} 
Solving this for $\widetilde{z} = (\widetilde{z}_1,\widetilde{z}_2,\widetilde{z}_3)$ in terms of $z=(z_1,z_2,z_3)$, we obtain the following map, which we call dET:
\begin{equation}\label{ET explicit inhom}
\left\lbrace \begin{aligned} \widetilde{z}_1 &= \frac{z_1 + 2 \varepsilon \alpha_1 z_2 z_3 + \varepsilon^2 z_1 (-\alpha_2 \alpha_3 z_1^2 + \alpha_3 \alpha_1 z_2^2 + \alpha_1 \alpha_2 z_3^2)}{\Delta(z,\varepsilon)} ,\\ 
\widetilde{z}_2 &= \frac{z_2 + 2\varepsilon \alpha_2 z_3 z_1 + \varepsilon^2 z_2 (\alpha_2 \alpha_3 z_1^2 - \alpha_3\alpha_1 z_2^2 + \alpha_1 \alpha_2 z_3^2)}{\Delta(z,\varepsilon)}, \\ 
\widetilde{z}_3 &= \frac{z_3 + 2\varepsilon \alpha_3 z_1 z_2 + \varepsilon^2 z_3(\alpha_2 \alpha_3 z_1^2 + \alpha_3 \alpha_1 z_2^2 - \alpha_1 \alpha_2 z_3^2)}{\Delta(z,\varepsilon)},  \end{aligned} \right. 
\end{equation}
where
\begin{equation}\label{dET denom}
\Delta(z,\varepsilon) = 1 - \varepsilon^2(\alpha_2 \alpha_3 z_1^2 + \alpha_3 \alpha_1 z_2^2 + \alpha_1 \alpha_2z_3^2) - 2 \varepsilon^3 \alpha_1 \alpha_2 \alpha_3 z_1 z_2 z_3.
\end{equation} 
Various aspects of integrability of dET were discussed in \cite{PS2010}, \cite{PPS2011}, \cite{Kimura_2017}. In particular, it possesses the following integrals of motion:
   \begin{equation}
   \mathcal{H}_1(\varepsilon) = \dfrac{\alpha_2 z_3^2 - \alpha_3 z_2^2}{1-\varepsilon^2 \alpha_2 \alpha_3 z_1^2},\quad
   \mathcal{H}_2(\varepsilon) = \dfrac{\alpha_3 z_1^2 - \alpha_1 z_3^2}{1-\varepsilon^2 \alpha_3 \alpha_1 z_2^2},\quad
   \mathcal{H}_3(\varepsilon) = \dfrac{\alpha_1 z_2^2 - \alpha_2 z_1^2}{1-\varepsilon^2 \alpha_1 \alpha_2 z_3^2}.
   \label{HdET}
   \end{equation}
Only two of them are independent, because of
\begin{equation}
\alpha_1 \mathcal{H}_1(\varepsilon) + \alpha_2 \mathcal{H}_2(\varepsilon) + \alpha_3 \mathcal{H}_3(\varepsilon) + \varepsilon^4 \alpha_1 \alpha_2 \alpha_3 \mathcal{H}_1(\varepsilon) \mathcal{H}_2(\varepsilon) \mathcal{H}_3(\varepsilon) =0.
\label{epsrel}
\end{equation}
Note that both the conserved quantities \eqref{HdET} and the relation \eqref{epsrel} are  $\varepsilon$-deformations of the corresponding objects for the continuous-time system.

In homogeneous coordinates $x=[x_1:x_2:x_3:x_4]$, we arrive at the degree 3 birational map $\widetilde{x}=f(x;\varepsilon)$ on $\mbP^3$:
\begin{equation} \label{dET hom}
\left\lbrace \begin{aligned}
\widetilde{x}_1 &= x_1 x_4^2 + 2 \varepsilon \alpha_1 x_2 x_3 x_4 + \varepsilon^2 x_1 (-\alpha_2 \alpha_3 x_1^2 + \alpha_3 \alpha_1 x_2^2 + \alpha_1 \alpha_2 x_3^2), \\
\widetilde{x}_2 &= x_2 x_4^2 + 2\varepsilon \alpha_2 x_3 x_1 x_4 + \varepsilon^2 x_2 (\alpha_2 \alpha_3 x_1^2 - \alpha_3\alpha_1 x_2^2 + \alpha_1 \alpha_2 x_3^2), \\
\widetilde{x}_3 &= x_3 x_4^2 + 2\varepsilon \alpha_3 x_1 x_2 x_4 + \varepsilon^2 x_3(\alpha_2 \alpha_3 x_1^2 + \alpha_3 \alpha_1 x_2^2 - \alpha_1 \alpha_2 x_3^2) ,\\
\widetilde{x}_4 &= x_4^3 - \varepsilon^2x_4(\alpha_2 \alpha_3 x_1^2 + \alpha_3 \alpha_1 x_2^2 + \alpha_1 \alpha_2x_3^2)  - 2 \varepsilon^3 \alpha_1 \alpha_2 \alpha_3 x_1 x_2 x_3.
\end{aligned} \right.
 \end{equation} 
Observe that the relations $\mathcal H_3=\mu$ and $\mathcal H_2=\lambda$ define the \emph{separable pencils of quadrics}
$$ 
\cP_\mu: \quad (\alpha_1 x_2^2 - \alpha_2 x_1^2)-\mu(x_4^2 - \varepsilon^2 \alpha_1 \alpha_2 x_3^2)=X_1X_2-\mu X_3X_4=0,
$$ 
and
$$
\cQ_\lambda: \quad (\alpha_3 x_1^2  - \alpha_1 x_3^2)-\lambda(x_4^2 - \varepsilon^2 \alpha_3 \alpha_1 x_2^2)=U_1U_2-\lambda U_3U_4=0,
$$
where we can choose the corresponding linear forms as follows:
\begin{equation}\label{XUdET}
\left\lbrace
\begin{aligned}
X_1 &= \sqrt{\alpha_1} x_2 - \sqrt{\alpha_2}x_1,\\
X_2 &= \sqrt{\alpha_1} x_2 + \sqrt{\alpha_2}x_1,\\
X_3 &= x_4 - \varepsilon \sqrt{\alpha_1\alpha_2}x_3,\\
X_4 &= x_4 + \varepsilon \sqrt{\alpha_1\alpha_2}x_3,
\end{aligned}
\right. \qquad \qquad 
\left\lbrace
\begin{aligned}
U_1 &= \sqrt{\alpha_3} x_1 - \sqrt{\alpha_1}x_3, \\
U_2 &= \sqrt{\alpha_3} x_1 + \sqrt{\alpha_1}x_3, \\
U_3 &=  x_4 - \varepsilon \sqrt{\alpha_3\alpha_1}x_2,\\
U_4 &=  x_4 + \varepsilon \sqrt{\alpha_3\alpha_1}x_2.
\end{aligned}
\right.
\end{equation}
Observe that both pencils of quadrics are invariant under $\sigma: X_3\leftrightarrow X_4$ which in coordinates $x$ reads as $\sigma: x_3\leftrightarrow -x_3$.
\smallskip

\begin{theo} \label{Theorem dET}
The linear forms \eqref{XUdET} satisfy the conditions of Theorem \ref {geprop}. The map 
\begin{equation}
f=\sigma\circ i_1=i_2\circ \sigma,
\label{th:eq}
\end{equation}
where $i_1$, $i_2$ are the involutions along the generators of $\cP_\mu=\{X_1X_2-\mu X_3X_4=0\}$ defined by the intersections with the pencil $\cQ_\lambda=\{U_1U_2-\lambda U_3U_4=0\}$, coincides with dET, when expressed in coordinates $x$.
\end{theo}
\begin{proof}
We express the linear forms $U_i$ in coordinates $X$:
\begin{eqnarray}
U_1 & = & \frac{1}{2\varepsilon\sqrt{\alpha_2}}\big(X_3-X_4+\varepsilon\sqrt{\alpha_3}(X_2-X_1)\big), \label{U1 dET}\\
U_2 & = & \frac{1}{2\varepsilon\sqrt{\alpha_2}}\big(X_4-X_3+\varepsilon\sqrt{\alpha_3}(X_2-X_1)\big), \label{U2 dET}\\
U_3 & = & \frac{1}{2}\big(X_3+X_4-\varepsilon\sqrt{\alpha_3}(X_1+X_2)\big),   \label{U3 dET} \\
U_4 & = & \frac{1}{2}\big(X_3+X_4+\varepsilon\sqrt{\alpha_3}(X_1+X_2)\big).  \label{U4 dET}
\end{eqnarray}
This allows us to easily compute equations of the lines $L_5,\ldots,L_8$ in coordinates $X$:
\begin{eqnarray}
L_5 & = & \{U_1 = U_3 = 0\}=\{ X_3-\varepsilon\sqrt{\alpha_3}X_1=0, \; X_4-\varepsilon\sqrt{\alpha_3}X_2=0\}, \\
L_6 & = & \{U_1 = U_4 =0\}= \{ X_3+\varepsilon\sqrt{\alpha_3}X_2=0, \; X_4+\varepsilon\sqrt{\alpha_3}X_1=0\}, \\
L_7 & = & \{U_2 = U_4 = 0\}=\{ X_3+\varepsilon\sqrt{\alpha_3}X_1=0, \; X_4+\varepsilon\sqrt{\alpha_3}X_2=0\}, \\
L_8 & = &  \{U_2 = U_3 =0\}= \{ X_3-\varepsilon\sqrt{\alpha_3}X_2=0, \; X_4-\varepsilon\sqrt{\alpha_3}X_1=0\}. 
\end{eqnarray}
Now one immediately checks that the four lines $L_2=\{X_1=X_4=0\}$, $L_4=\{X_2=X_3=0\}$, $L_5$ and $L_7$ are pairwise skew and lie on the quadric $A=0$, where
\begin{equation}\label{A dET}
A=X_3X_4-\varepsilon^2\alpha_3X_1X_2.
\end{equation}
Thus, the conditions of Theorem  \ref {geprop} (case (b)) are satisfied and it follows that the map $f=\sigma\circ i_1=i_2\circ \sigma$ has the form
\begin{equation}\label{dET in X}
f: [X_1: X_2: X_3 : X_4] \mapsto [ X_1B_2:  X_2 B_0: X_4B_2 : X_3B_0],
\end{equation} 
where 
\begin{equation}\label{B dET}
B_0 =  X_4^2 - \varepsilon^2\alpha_3X_1^2,\qquad B_2 = X_3^2 - \varepsilon^2\alpha_3X_2^2.
\end{equation}
As guaranteed by Theorem  \ref {geprop}, $B_0$ vanishes on $L_2$, $L_6$ and $L_8$, while $B_2$ vanishes on $L_4$, $L_6$ and $L_8$.

Now it is a matter of a straightforward computation to see that,  in the coordinates $x$ given by  \eqref{XUdET}, the map \eqref{dET in X}  coincides with the map  \eqref{dET hom}.
\end{proof}

The decompositions \eqref{th:eq} are illustrated in figure \ref{fig}.

\begin{figure}[h]
\begin{center}
\includegraphics[trim=0 2cm 1cm 3cm , scale=0.34]{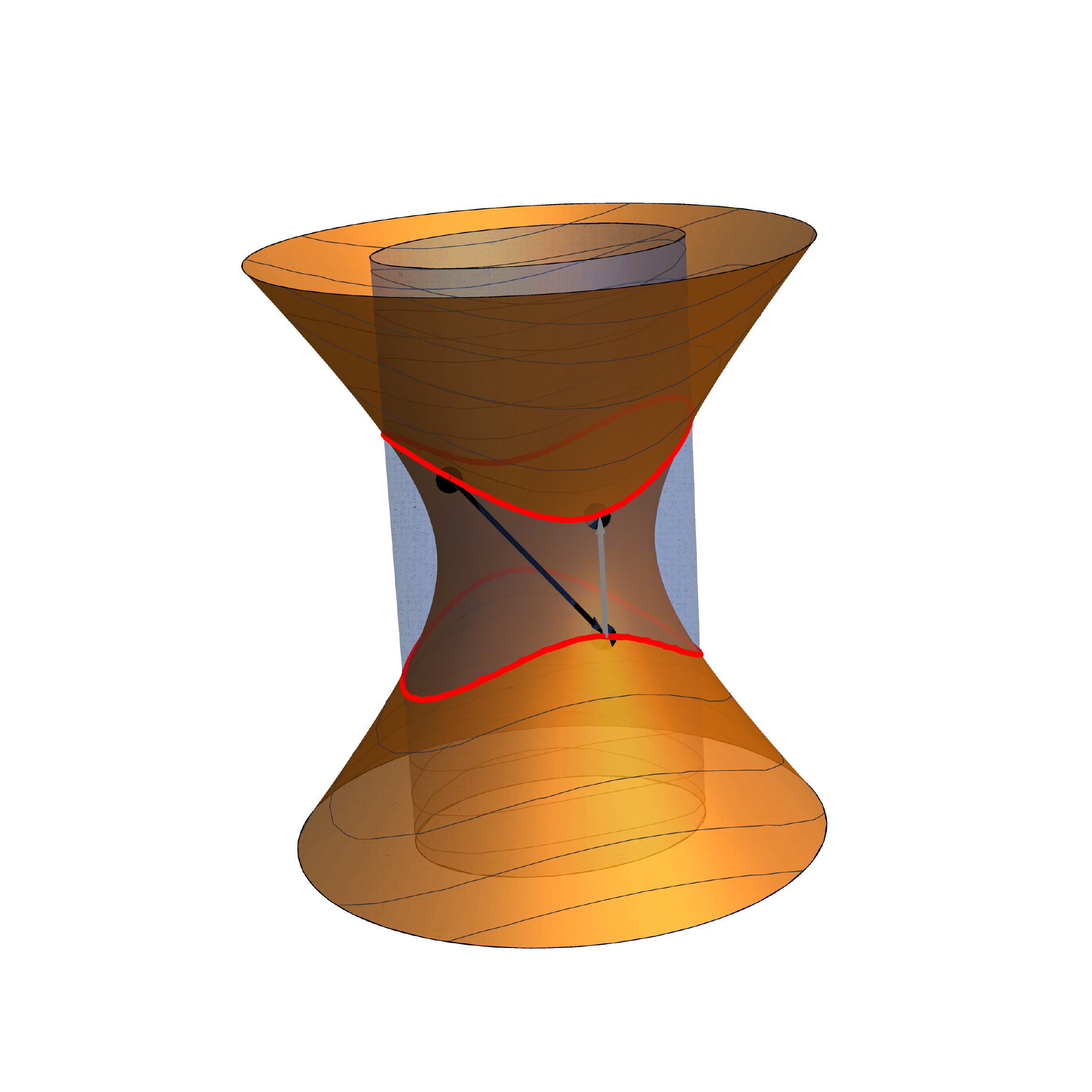}\includegraphics[trim=1cm 2cm 0 3cm , scale=0.34]{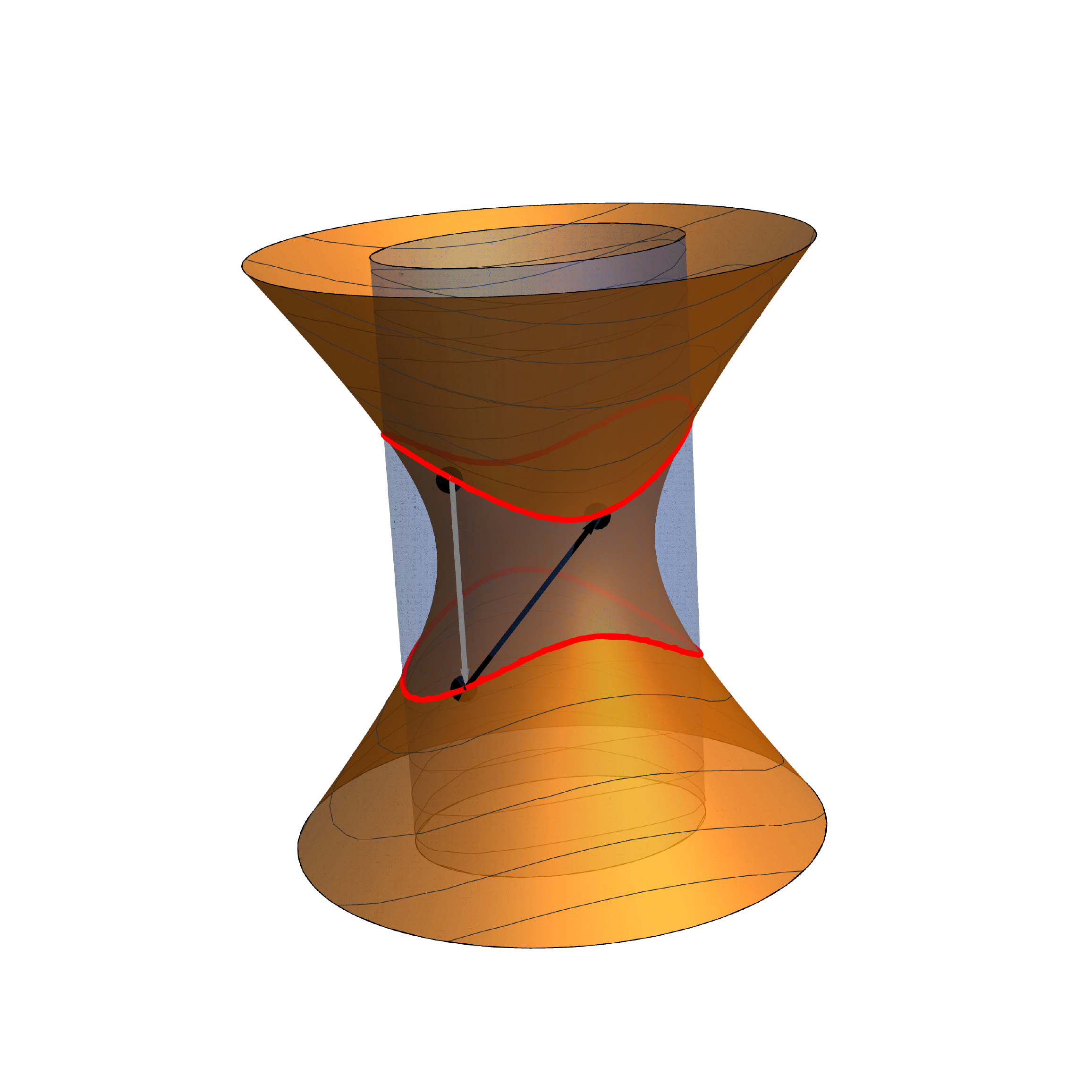}
\end{center}
\caption{Map dET as a composition of two involutions, $f=\sigma\circ i_1$ (left) and $f=i_2 \circ \sigma$ (right).
}
\label{fig}
\end{figure}

\textbf{Remark 1.} The factorization of dET in a composition of involutions along generators of a quadric was first considered in thesis \cite{Smeenk} under the guidance of the second author.

\textbf{Remark 2.} From \eqref{U1 dET}--\eqref{U4 dET} one sees that the quadric $\{A=0\}$ with $A$ from \eqref{A dET} belongs to both pencils $\cP_\mu=\{X_1X_2-\mu X_3X_4=0\}$ and $\cQ_\lambda=\{U_1U_2-\lambda U_3U_4=0\}$. Thus, dET can be also seen as a particular case of the construction of Section \ref{section: construction 2}. As stated in Theorem \ref{Theorem construction 2},
both  quadrics $\{B_0=0\}$ and $\{B_2=0\}$ degenerate into a pair of planes. Their intersection consists of the lines $L_6$, $L_8$ and 
\begin{eqnarray}
L_9 & = & \{ X_3+\varepsilon\sqrt{\alpha_3}X_2=0, \; X_4-\varepsilon\sqrt{\alpha_3}X_1=0\}, \\ 
L_{10} & = & \{ X_3-\varepsilon\sqrt{\alpha_3}X_2=0, \; X_4+\varepsilon\sqrt{\alpha_3}X_1=0\}.
\end{eqnarray}
The intersection points
$$
L_6\cap L_9=[0:1:-\varepsilon\sqrt{\alpha_3}:0], \qquad L_8\cap L_{10} = [0:1:\varepsilon\sqrt{\alpha_3}:0]
$$
lie on the line $L_2$, while the intersection points
$$
L_6\cap L_{10}=[1:0:0:-\varepsilon\sqrt{\alpha_3}], \qquad L_8\cap L_9 = [1:0:0:\varepsilon\sqrt{\alpha_3}]
$$
lie on the line $L_4$. Thus, the six lines 
$$
 I(f) = L_2 \cup L_4 \cup L_6 \cup L_8 \cup L_9 \cup L_{10},
$$ 
constitute the side lines of a tetrahedron. Birational determinantal maps of $\mbP^3$ with the indeterminacy set of this (tetrahedron) type are well known in the classical literature, see, e.g., \cite{Noether}, \cite{Hudson2}, \cite{Hudson1927}. They can be represented as $M_2\circ i\circ M_1$, where $M_1$, $M_2$ are linear projective maps, while 
\begin{eqnarray*}
i(Y_1,Y_2,Y_3,Y_4) & = & [Y_2Y_3Y_4 : Y_1Y_3Y_4:Y_1Y_2Y_4:Y_1Y_2Y_3]\\
& = & [1/Y_1:1/Y_2:1/Y_3:1/Y_4]
\end{eqnarray*} 
is the standard cubic inversion involution on $\mbP^3$. One can now easily find $M_1$, $M_2$ for dET in both coordinate systems $x$ and $X$. In particular, in coordinates $x$ we have:
\begin{pro} Map \eqref{dET hom}  coincides with $M_2\circ i\circ M_1$, where $M_1$, $M_2$ are linear projective maps with the matrices
\begin{equation}
M_1=\begin{pmatrix} -b_1 & b_2 & b_3 & 1\\ b_1 & -b_2 & b_3 & 1\\ b_1 & b_2 & -b_3 & 1\\ -b_1 & -b_2 & -b_3 & 1\end{pmatrix},  \quad 
M_2^{-1}=\begin{pmatrix} b_1 & -b_2 & -b_3 & 1\\ -b_1 & b_2 & -b_3 & 1\\ -b_1 & -b_2 & b_3 & 1\\ b_1 & b_2 & b_3 & 1\end{pmatrix}.  
\end{equation}
where
\begin{equation}
b_1=\epsilon\sqrt{\alpha_2\alpha_3}, \quad b_2=\epsilon\sqrt{\alpha_1\alpha_3}, \quad b_3=\epsilon\sqrt{\alpha_1\alpha_2}.
\end{equation}
\end{pro}

\clearpage

\section{Example II: Kahan discretization of the Zhukovski-Volterra gyrostat with one non-vanishing $\beta_k$}
\label{sect: dZV beta1}

The Zhukovski-Volterra gyrostat is a generalization of the Euler top:
\begin{equation} \label{ZV}
	\begin{cases}
			\dot{z}_1 = \alpha_1 z_2 z_3 + \beta_3 z_2 - \beta_2 z_3, \\
			\dot{z}_2 = \alpha_2 z_3 z_1 + \beta_1 z_3 - \beta_3 z_1,\\
			\dot{z}_3 = \alpha_3 z_1 z_2 + \beta_2 z_1 - \beta_1 z_2.
	\end{cases} 
\end{equation}	
Here, $(\beta_1,\beta_2,\beta_3)$ represents the vector of the gyrostatic momentum. This system is integrable if $ \alpha_1 + \alpha_2 + \alpha_3=0$, with integrals of motion
\begin{eqnarray}
	H_1 & = & \alpha_2 z_3^2 - \alpha_3 z_2^2 - 2(\beta_1 z_1 + \beta_2 z_2 + \beta_3 z_3), \\
	H_2 & = & \alpha_3 z_1^2 - \alpha_1 z_3^2 - 2(\beta_1 z_1 + \beta_2 z_2 + \beta_3 z_3), \\
	H_3 & = & \alpha_1 z_2^2 - \alpha_2 z_1^2 - 2(\beta_1 z_1 + \beta_2 z_2 + \beta_3 z_3).
\end{eqnarray}
Only two of them are independent, due to the relation $\alpha_1 H_1 + \alpha_2 H_2 + \alpha_3 H_3 =0$, which holds true provided $\alpha_1 + \alpha_2 + \alpha_3 =0$. 

Integrability of the Kahan discretization of Zhukovski-Volterra gyrostat was studied in \cite{PPS2011}. In the present section, we give a geometric interpretation of their result concerning another case in which the system is integrable, namely $\beta_2 = \beta_3 =0$, which we denote by ZV($\beta_1$). Equations of motion simplify to
\begin{equation} \label{ZV beta1}
	\begin{cases}
			\dot{z}_1 = \alpha_1 z_2 z_3, \\
			\dot{z}_2 = \alpha_2 z_3 z_1 + \beta_1 z_3,\\
			\dot{z}_3 = \alpha_3 z_1 z_2 - \beta_1 z_2,
	\end{cases} 
\end{equation}	
and are integrable without further restrictions on parameters. More precisely, the following two functions are integrals of motion of \eqref{ZV beta1} for arbitrary values of parameters:
\begin{equation}
H_2 = \alpha_3 z_1^2 - \alpha_1 z_3^2 - 2\beta_1 z_1, \qquad H_3 = \alpha_1 z_2^2 - \alpha_2 z_1^2 - 2\beta_1 z_1.
\end{equation}
The Kahan discretization of \eqref{ZV beta1} is given by the implicit equations of motion:
\begin{equation} \label{defdZV}
\begin{cases}
			\widetilde{z}_1-z_1 = \varepsilon \alpha_1 (\widetilde{z}_2 z_3 + z_2 \widetilde{z}_3), \\
			\widetilde{z}_2-z_2 = \varepsilon \alpha_2 (\widetilde{z}_3 z_1 + z_3 \widetilde{z}_1) + \varepsilon \beta_1 (\widetilde{z}_3 + z_3), \\
			\widetilde{z}_3-z_3 = \varepsilon \alpha_3 (\widetilde{z}_1 z_2 + z_1 \widetilde{z}_2) - \varepsilon \beta_1 (\widetilde{z}_2 + z_2).
			\end{cases}
\end{equation}
This defines a birational map $\widetilde{z} = f(z;\varepsilon)$ which we will denote by dZV$(\beta_1)$. It was found in \cite{PPS2011} that this map has two conserved quantities:
\begin{equation}\label{HdZV}
\mathcal{H}_2(\varepsilon) = \dfrac{\alpha_3 z_1^2 - \alpha_1 z_3^2 - 2\beta_1 z_1+ \dfrac{\beta_1^2}{\alpha_3}}{1-\varepsilon^2 \alpha_3 \alpha_1 z_2^2}, \quad 
\mathcal{H}_3(\varepsilon) = \dfrac{\alpha_1 z_2^2 - \alpha_2 z_1^2 - 2\beta_1 z_1 - \dfrac{\beta_1^2}{ \alpha_2}}{1-\varepsilon^2 \alpha_1 \alpha_2 z_3^2}
\end{equation}
The same name will be used for the corresponding degree 3 birational map $\widetilde{x}=f(x;\varepsilon)$ on $\mbP^3$, expressed in homogeneous coordinates $x=[x_1:x_2:x_3:x_4]$.

Each of the relations $\mathcal H_3=\mu$ and $\mathcal H_2=\lambda$ defines a \emph{separable pencil of quadrics},
$$ 
\cP_\mu: \quad \Big(\alpha_1 x_2^2 - \alpha_2 x_1^2- 2\beta_1 x_1x_4 - \dfrac{\beta_1^2}{ \alpha_2}x_4^2\Big)-\mu(x_4^2 - \varepsilon^2 \alpha_1 \alpha_2 x_3^2)=X_1X_2-\mu X_3X_4=0,
$$ 
resp.
$$
\cQ_\lambda: \quad \Big(\alpha_3 x_1^2  - \alpha_1 x_3^2- 2\beta_1 x_1x_4+ \dfrac{\beta_1^2}{\alpha_3}x_4^2\Big)-\lambda(x_4^2 - \varepsilon^2 \alpha_3 \alpha_1 x_2^2)=U_1U_2-\lambda U_3U_4=0,
$$
where we can choose the corresponding linear forms as follows:
\begin{equation} \label{XUdZV}
\begin{cases}
	X_1 = \sqrt{\alpha_1} x_2 - \sqrt{\alpha_2} x_1 - \dfrac{\beta_1}{\sqrt{\alpha_2}} x_4,\\
	X_2 = \sqrt{\alpha_1} x_2 + \sqrt{\alpha_2} x_1 + \dfrac{\beta_1}{\sqrt{\alpha_2}} x_4,\\
	X_3 = x_4 - \varepsilon \sqrt{\alpha_1 \alpha_2} x_3,\\
	X_4 = x_4 + \varepsilon \sqrt{\alpha_1 \alpha_2} x_3,
	\end{cases}
	\qquad
	\begin{cases}
	U_1 =  \sqrt{\alpha_3} x_1 - \sqrt{\alpha_1} x_3- \dfrac{\beta_1}{\sqrt{\alpha_3}} x_4,\\
	U_2 = \sqrt{\alpha_3} x_1  + \sqrt{\alpha_1} x_3- \dfrac{\beta_1}{\sqrt{\alpha_3}} x_4,\\
	U_3 = x_4 - \varepsilon \sqrt{\alpha_1 \alpha_3} x_2,\\
	U_4 = x_4 + \varepsilon \sqrt{\alpha_1 \alpha_3} x_2.
	\end{cases}
\end{equation}
	
Also in the present case, both pencils of quadrics are invariant under $\sigma: X_3\leftrightarrow X_4$ which in coordinates $x$ reads as $\sigma: x_3\leftrightarrow -x_3$.
\smallskip

\begin{theo} \label{Theorem dZV(beta1)}
The linear forms \eqref{XUdZV} satisfy the conditions of Theorem \ref {geprop}. The map 
$$
f=\sigma\circ i_1=i_2\circ \sigma,
$$ 
coincides with dZV$(\beta_1)$ when expressed in coordinates $x$, where $i_1$, $i_2$ are the involutions along the generators of $\cP_\mu=\{X_1X_2-\mu X_3X_4=0\}$ defined by the intersections with the pencil $\cQ_\lambda=\{U_1U_2-\lambda U_3U_4=0\}$.
\end{theo}
\begin{proof}
The computations go along the same lines as in the proof of Theorem \ref{Theorem dET}. Equation $\{A=0\}$ of the quadric containing the lines $L_2,L_4, L_5$ and $L_7$ reads:
\begin{equation}
A =  X_3X_4-\varepsilon^2\alpha_3X_1X_2+\varepsilon^2\gamma\sqrt{\alpha_3}(X_1X_3-X_2X_4),
\end{equation}
where
\begin{equation}
\gamma=\beta_1 \frac{\alpha_2+\alpha_3}{2\sqrt{\alpha_2\alpha_3}}.
\end{equation}
The map $f$ is of the form \eqref{dET in X} with
\begin{eqnarray}
B_2 & = & X_3^2-\varepsilon^2\alpha_3X_2^2+\varepsilon\gamma\sqrt{\alpha_3}(X_2X_4+X_1X_3+2X_2X_3), \\
B_0 & = & X_4^2-\varepsilon^2\alpha_3X_1^2-\varepsilon\gamma\sqrt{\alpha_3}(X_2X_4+X_1X_3+2X_1X_4).
\end{eqnarray}
A straightforward computation shows that in coordinates $x$ this map coincides with dZV$(\beta_1)$.
\end{proof}

\textbf{Remark.}  For the map dZV$(\beta_1)$, the structure of the indeterminacy set is as in Proposition \ref{prop indeterminacy set} (case (b)). Thus, it belongs to the class of birational maps introduced by Cayley in \cite[${\rm n}^{\rm o}$ 102--104]{Cayley}, see also \cite[Example A1]{Hudson2}.

\section{Example III: Kahan-type discretization of a \\ special Zhukovski-Volterra gyrostat with two \\ non-vanishing $\beta_k$}
\label{ss: dZV special}

We now turn to the problem of an integrable discretization of the Zhukovski-Volterra gyrostat when $\beta_3=0$, which we denote by ZV($\alpha_1,\alpha_2,\alpha_3,\beta_1,\beta_2$):
\begin{equation}\label{ZV beta 2 beta 3}
\begin{cases}
	\dot{z}_1 = \alpha_1 z_2 z_3 - \beta_2 z_3, \\
	\dot z_2 = \alpha_2 z_3 z_1 + \beta_1 z_3, \\
	\dot z_3 = \alpha_3 z_1 z_2 +\beta_2 z_1 - \beta_1 z_2.
\end{cases}
\end{equation}
One can easily check that the function $H_3=\alpha_1 z_2^2 - \alpha_2 z_1^2 - 2(\beta_1 z_1 + \beta_2 z_2)$ is an integral of motion for arbitrary values of parameters, while under the condition $\alpha_1+\alpha_2+\alpha_3=0$ the system acquires the second integral of motion $H_2=\alpha_3 z_1^2 - \alpha_1 z_3^2 - 2(\beta_1 z_1 + \beta_2 z_2)$. Thus, integrability take place under the above mentioned condition only.

The Kahan discretization of this system, denoted by dZV($\alpha_1,\alpha_2,\alpha_3,\beta_1,\beta_2$), is defined by implicit equations of motion
$$
\begin{cases}
	\widetilde{z}_1-z_1 = \varepsilon \alpha_1 (\widetilde{z}_2 z_3 + z_2 \widetilde{z}_3)- \varepsilon \beta_2 (\widetilde{z}_3 + z_3), \\
	\widetilde{z}_2-z_2 = \varepsilon \alpha_2 (\widetilde{z}_3 z_1 + z_3 \widetilde{z}_1) + \varepsilon \beta_1 (\widetilde{z}_3 + z_3), \\
	\widetilde{z}_3-z_3 = \varepsilon \alpha_3 (\widetilde{z}_1 z_2 + z_1 \widetilde{z}_2) + \varepsilon \beta_2 (\widetilde{z}_1 + z_1)- \varepsilon \beta_1 (\widetilde{z}_2 + z_2).
\end{cases}
$$
The corresponding birational map $\widetilde{z} = f(z,\varepsilon)$ has, for arbitrary values of parameters, one conserved quantity: 
\begin{equation} \label{dZV H3 gen}
\mathcal{H}_3(\varepsilon) = \dfrac{\alpha_1 z_2^2 - \alpha_2 z_1^2 - 2(\beta_1 z_1 + \beta_2 z_2) +\dfrac{\beta_2^2}{\alpha_1} -\dfrac{\beta_1^2}{\alpha_2}}{1-\varepsilon^2 \alpha_1 \alpha_2 z_3^2}.
\end{equation}
However, it does not possess the second one, even under the condition $\alpha_1+\alpha_2+\alpha_3=0$. In \cite{PPS2011}, a particular case was identified, namely $\alpha_1 = - \alpha_2 = \alpha$, for which the map $f$ admits the second integral of motion. The additional integral of dZV$(\alpha,-\alpha,0,\beta_1,\beta_2)$ is polynomial and reads:
\begin{equation} 
\mathcal{H}_2(\varepsilon) = -\alpha z_3^2 - 2(\beta_1 z_1 + \beta_2 z_2) + \varepsilon^2 \alpha (\beta_2 z_1 - \beta_1 z_2)^2.
\end{equation}
We observe that, while the pencil of quadrics in $\mbP^3$ corresponding to $\mathcal H_3(\varepsilon)=\mu$ is separable, this is not the case for the pencil of quadrics corresponding to $\mathcal H_2(\varepsilon)=\lambda$. Indeed, the latter does not contain two pairs of distinct planes, but rather one double plane at infinity $\{x_4^2=0\}$, and its base set consists of two double lines. Thus, the map dZV$(\alpha,-\alpha,0,\beta_1,\beta_2)$ apparently is not covered by our constructions.

We now present a novel one-parameter family of discretizations of the special Zhukovski-Volterra gyrostat ZV$(\alpha,-\alpha,0,\beta_1,\beta_2)$, based on the construction with two separable pencils, for which the map dZV$(\alpha,-\alpha,0,\beta_1,\beta_2)$ is a special (or, better, a limiting) case.

\begin{theo} \label{th:newsys1}
Consider the following linear forms:
\begin{equation}
\begin{cases}
X_1 = \sqrt{\alpha}(x_1 + ix_2) - \dfrac{1}{\sqrt{\alpha}}(\beta_1+i\beta_2) x_4, \\[0.3cm]
X_2 = \sqrt{\alpha}(x_1 - ix_2) - \dfrac{1}{\sqrt{\alpha}}(\beta_1-i\beta_2) x_4,\\[0.2cm]
X_3 = x_4 - i\varepsilon \alpha x_3,\\[0.2cm]
X_4 = x_4 + i \varepsilon \alpha x_3,
\end{cases}
\end{equation}
and
\begin{equation}
\left\lbrace
\begin{aligned}
U_1 =& \big(1-\varepsilon^2\delta^2(\beta_1^2+\beta_2^2)\big)x_4+\varepsilon^2 \delta^2 \alpha(\beta_1x_1+\beta_2x_2)+ 
\varepsilon^2 \delta \alpha(\beta_2x_1-\beta_1x_2),  \\[0.1cm]
U_2 =& \big(1-\varepsilon^2\delta^2(\beta_1^2+\beta_2^2)\big)x_4+\varepsilon^2 \delta^2 \alpha(\beta_1x_1+\beta_2x_2)- 
\varepsilon^2 \delta \alpha(\beta_2x_1-\beta_1x_2), \\[0.1cm]
U_3 =&\; x_4 +\varepsilon\delta\alpha x_3, \\[0.1cm]
U_4 =&\; x_4 - \varepsilon\delta\alpha x_3.
\end{aligned}
\right.
\end{equation}
These forms satisfy the conditions of Theorem \ref{geprop}. The map 
$$
f=\sigma\circ i_1=i_2\circ \sigma,
$$ 
where $i_1$, $i_2$ are the involutions along the generators of $\cP_\mu=\{X_1X_2-\mu X_3X_4=0\}$ defined by the intersections with the pencil $\cQ_\lambda=\{U_1U_2-\lambda U_3U_4=0\}$, and $\sigma$ is the involution $x_3\leftrightarrow -x_3$, or $X_3\leftrightarrow X_4$, is given in the affine chart $[z_1:z_2:z_3:1]$ of the coordinate system $x$ by the following implicit equations of motion, namely:
\begin{equation}\label{dZV spec}
\left\{ \begin{aligned}
\widetilde{z}_1-z_1 &= \varepsilon \alpha (\widetilde{z}_2 z_3 + z_2 \widetilde{z}_3)- \varepsilon \beta_2 (\widetilde{z}_3 + z_3), \\
\widetilde{z}_2-z_2 &= -\varepsilon \alpha (\widetilde{z}_3 z_1 + z_3 \widetilde{z}_1) + \varepsilon \beta_1 (\widetilde{z}_3 + z_3), \\
\widetilde{z}_3-z_3 &= 
\dfrac{ \varepsilon\beta_2 (\zti_1 + z_1) -  \varepsilon\beta_1 (\zti_2 + z_2) +\varepsilon^2\delta^2\alpha\big(\beta_1 (\zti_1 z_3-\zti_3 z_1 )
+\beta_2(\zti_2 z_3 - \zti_3 z_2)\big)}{1- \varepsilon^2\delta^2 (\beta_1^2 + \beta_2^2)}.	
\end{aligned} \right.
\end{equation}
This map admits two integrals of motion:
\begin{equation} \label{dZV spec H3}
\mathcal H_3(\varepsilon)=\dfrac{\alpha(z_1^2+z_2^2)-2 (\beta_1 z_1  + \beta_2 z_2)+(\beta_1^2+\beta_2^2)/\alpha }
{1+\varepsilon^2\alpha^2  z_3^2}
\end{equation}
and 
\begin{equation} \label{dZV spec H2}
\mathcal H_2(\varepsilon,\delta)=\frac{ \big(1-\varepsilon^2\delta^2(\beta_1^2+\beta_2^2)+\varepsilon^2 \delta^2 \alpha(\beta_1z_1+\beta_2z_2)\big)^2-
\varepsilon^4 \delta^2 \alpha^2(\beta_2z_1-\beta_1z_2)^2}{1-\varepsilon^2\delta^2\alpha^2 z_3^2}.
\end{equation}	
\end{theo}
\begin{proof}
This is a straightforward computation along the same lines as the proof of Theorem \ref{Theorem dET}. The integrals of motion are just 
$$
\mathcal H_3(\varepsilon)= \dfrac{X_1 X_2}{X_3 X_4} \quad \text{and} \quad \mathcal H_2(\varepsilon,\delta)=\dfrac{U_1 U_2}{U_3 U_4}. 
 $$ 
\end{proof}

\textbf{Remark.} We notice that \eqref{dZV spec} is not a Kahan discretization of ZV$(\alpha,-\alpha,0,\beta_1,\beta_2)$ in the strict sense, because of the presence of skew-symmetric bilinear expressions $\zti_1 z_3-\zti_3 z_1$ and $\zti_2 z_3-\zti_3 z_2$ on the right-hand side of the third equation of motion. However, these terms do not contribute towards the continuous limit $\varepsilon\to 0$, so that for any $\delta$ we get an integrable discretization of ZV$(\alpha,-\alpha,0,\beta_1,\beta_2)$. We can speak in this case of an \emph{adjusted} Kahan discretization, in the sense of \cite{PSZ2020}, \cite{SST2021}. In the limit $\delta\to 0$, we recover the map dZV$(\alpha,-\alpha,0,\beta_1,\beta_2)$. The second integral of the latter map is recovered in this limit, as well, due to
$$
\mathcal H_2(\varepsilon,\delta)=1-\varepsilon^2\delta^2\alpha \mathcal H_2(\epsilon) +O(\delta^4).
$$
On the other hand, if $\delta^2=-1$, so that the integrals $\mathcal H_3(\varepsilon)$ and $\mathcal H_2(\varepsilon,\delta)$ share the common denominator, then their linear combination leads to a simpler version of the second integral, namely
\begin{equation} \label{dZV spec H2 spec}
\mathcal H_2(\varepsilon)=\dfrac{-\alpha z_3^2-2 (\beta_1 z_1+\beta_2 z_2)+2(\beta_1^2+\beta_2^2)/\alpha} 
{1+\varepsilon^2\alpha^2  z_3^2}.
\end{equation}

\section{Example IV: Kahan-type discretization of a \\
general Zhukovski-Volterra gyrostat with two \\ non-vanishing $\beta_k$}
\label{ss: dZV general}

Here, we give an application of the construction of Section \ref{section: construction 2}.

\begin{theo}
Define the following linear forms:
\begin{equation} \label{dZV gen X}
\begin{cases}
X_1 = \sqrt{\alpha_1}x_2 -\sqrt{\alpha_2}x_1 
-\Big(\dfrac{\beta_1}{\sqrt{\alpha_2}}+\dfrac{\beta_2}{\sqrt{\alpha_1}}\Big) x_4, \\[0.3cm]
X_2 = \sqrt{\alpha_1}x_2 +\sqrt{\alpha_2}x_1 
+\Big(\dfrac{\beta_1}{\sqrt{\alpha_2}}-\dfrac{\beta_2}{\sqrt{\alpha_1}}\Big) x_4,\\[0.2cm]
X_3 = x_4 - \varepsilon \sqrt{\alpha_1\alpha_2} x_3,\\[0.2cm]
X_4 = x_4 +  \varepsilon\sqrt{\alpha_1\alpha_2} x_3.
\end{cases}
\end{equation}
Set
\begin{equation}\label{dZV gen Qinfty}
Q_\infty(X)=\alpha_3 x_1^2 - \alpha_1 x_3^2 - 2(\beta_1 x_1 + \beta_2 x_2)x_4+\gamma x_4^2,
\end{equation}
where
\begin{equation}\label{dZV gen gamma}
\gamma=\frac{\beta_2^2}{\alpha_1}-\frac{\beta_1^2}{\alpha_2}
\end{equation}
(expressed in the variables $X$). Then the map 
$$
f=\sigma\circ i_1=i_2\circ \sigma,
$$ 
where $i_1$, $i_2$ are the involutions along the generators of $\cP_\mu=\{X_1X_2-\mu X_3X_4=0\}$ defined by the intersections with the pencil $\cQ_\lambda=\{X_3X_4-\lambda Q_\infty(X)=0\}$, and $\sigma$ is the involution $x_3\leftrightarrow -x_3$, or $X_3\leftrightarrow X_4$, is given in the coordinates $X$ by \eqref{i1 red}, where
\begin{equation}
B_2=Q_\infty(0,X_2,X_3,0),\quad 
B_0=Q_\infty(X_1,0,0,X_4).
\end{equation}
In the affine chart $[z_1:z_2:z_3:1]$ of the coordinate system $x$, the map $f$ is given by the following implicit equations of motion:
\begin{equation}\label{dZV gen}
\left\{ \begin{aligned}
\widetilde{z}_1-z_1 = & \;\varepsilon \alpha_1 (\widetilde{z}_2 z_3 + z_2 \widetilde{z}_3)- \varepsilon \beta_2 (\widetilde{z}_3 + z_3), \\
\widetilde{z}_2-z_2 = & \; \varepsilon \alpha_2 (\widetilde{z}_3 z_1 + z_3 \widetilde{z}_1) + \varepsilon \beta_1 (\widetilde{z}_3 + z_3), \\
\widetilde{z}_3-z_3 = & \; \varepsilon \alpha_3 (\widetilde{z}_1z_2+z_1\widetilde{z}_2)-\varepsilon\beta_2\dfrac{\alpha_2+\alpha_3}{\alpha_1}(\widetilde{z}_1+z_1)-\varepsilon\beta_1(\widetilde z_2+ z_2)\\
& -\varepsilon^2\beta_1(\alpha_2+\alpha_3) (z_1\widetilde{z}_3-\widetilde{z}_1z_3)-\varepsilon^2\beta_2\alpha_2(z_2\widetilde{z}_3-\widetilde{z}_2z_3).
\end{aligned} \right.
\end{equation}
This map possesses two integrals of motion, $\mathcal H_3(\varepsilon)$ given in \eqref{dZV H3 gen} and
\begin{equation}\label{dZV gen H2}
\mathcal H_2(\varepsilon)=\dfrac{\alpha_3 z_1^2 - \alpha_1 z_3^2 - 2(\beta_1 z_1 + \beta_2 z_2)+\gamma}{1-\varepsilon^2\alpha_1\alpha_2z_3^2}.
\end{equation}
\end{theo} 
\begin{proof}
The statement in coordinates $X$ follows from Theorem \ref{Theorem construction 2}. The result in coordinates $x$ follows by a direct symbolic computation. This computation is facilitated by a formulation of equations of motion in coordinates $X$ in a bilinear form. Let 
$$
Q_\infty=a_{11}X_1^2+a_{12}X_1X_2+a_{22}X_2^2+a_{13}X_1(X_3+X_4)+a_{23}X_2(X_3+X_4)+a_{33}(X_3^2+X_4^2)+a_{34}X_3X_4
$$
be a quadratic homogeneous polynomial symmetric w.r.t. $X_3\leftrightarrow X_4$, so that
$$
B_2=a_{22}X_2^2+a_{23}X_2X_3+a_{33}X_3^2, \qquad B_0=a_{11}X_1^2+a_{13}X_1X_4+a_{33}X_4^2.
$$
Then the relations
$$
[\widetilde X_1:\widetilde X_2:\widetilde X_3:\widetilde X_4]=[X_1B_2: X_2B_0:X_4B_2:X_3B_0]
$$
are equivalent to the system of bilinear relations between $X,\widetilde X$ of which three linearly independent ones can be chosen as follows:
$$
\widetilde X_1X_4=\widetilde X_3X_1, \qquad \widetilde X_2X_3=\widetilde X_4X_2, 
$$
$$
a_{11}\widetilde X_1X_1-a_{22}\widetilde X_2X_2+a_{13}\widetilde X_3X_1-a_{23}\widetilde X_4X_2+a_{33}(\widetilde X_3X_4-\widetilde X_4X_3)=0.
$$
Performing a linear change of variables according to \eqref{dZV gen X}, one finds three linearly independent bilinear relations between $x, \widetilde x$, which turn into \eqref{dZV gen} upon setting $z_i=x_i/x_4$ and $\widetilde z_i=\widetilde x_i/\widetilde x_4$.
\end{proof}

The map \eqref{dZV gen} is an ``adjusted'' Kahan-type discretization of the following system of differential equations:
\begin{equation}\label{ZV Nambu}
\left\{ \begin{aligned}
\dot{z}_1= & \; \alpha_1 z_2 z_3 -\beta_2 z_3, \\
\dot z_2 = & \; \alpha_2 z_3 z_1 + \beta_1 z_3, \\
\dot z_3 = & \; \alpha_3 z_1z_2-\beta_2\dfrac{\alpha_2+\alpha_3}{\alpha_1}  z_1-\beta_1z_2.
\end{aligned} \right.
\end{equation}
This system admits two conserved quantities $H_3=\alpha_1 z_2^2 - \alpha_2 z_1^2 - 2(\beta_1 z_1 + \beta_2 z_2)$ and $H_2=\alpha_3 z_1^2 - \alpha_1 z_3^2 - 2(\beta_1 z_1 + \beta_2 z_2)$ without any restrictions on parameters. Under condition $\alpha_1+\alpha_2+\alpha_3=0$, it turns into ZV$(\alpha_1,\alpha_2,\alpha_3,\beta_1,\beta_2)$, and \eqref{dZV gen}  turns into an integrable Kahan-type discretization of the latter system. If $\alpha_1=-\alpha_2=\alpha$  and $\alpha_3=0$, we recover the system ZV$(\alpha,-\alpha,0,\beta_1,\beta_2)$. If we choose in \eqref{dZV gen Qinfty} the value
\begin{equation}\label{dZV gen 2 gamma}
\gamma=2\left(\frac{\beta_2^2}{\alpha_1}-\frac{\beta_1^2}{\alpha_2}\right)
\end{equation}
instead of \eqref{dZV gen gamma}, we recover the discretization \eqref{dZV spec} of ZV$(\alpha,-\alpha,0,\beta_1,\beta_2)$ with $\delta^2=-1$ (note that the integral \eqref{dZV gen H2} with $\gamma$ from \eqref{dZV gen 2 gamma} coincides with the integral \eqref{dZV spec H2 spec}, if $\alpha_1=-\alpha_2=\alpha$ and $\alpha_3=0$).

\textbf{Remark.} System \eqref{ZV Nambu} can be interpreted as the Nambu system \cite{Nambu}
$$
\dot{z}=\frac{1}{4\alpha_1}\,\nabla H_2\times \nabla H_3.
$$ 
Some results on integrability of the Kahan discretization for Nambu systems were found in \cite{Celledoni2}, \cite{Celledoni3}. More precisely, in \cite{Celledoni2} integrability of the Kahan discretization was established for the case when both Nambu Hamiltonians are homogeneous quadratic polynomials on $\mbR^3$ (a typical example is given by dET). In \cite{Celledoni3}, for the case when both Nambu Hamiltonians are possibly inhomogeneous polynomials of degree 2 on $\mbR^3$, but each of them depends only on two of the three variables (a typical example being dZV$(\beta_1)$). Neither of these results covers our present case, where an adjustment of the Kahan discretization by means of nontrivial skew-symmetric bilinear forms of $z$, $\widetilde z$ is required.

\section{Conclusion}

In the present paper, we propose a geometric construction of three-dimensional birational maps preserving two pencils of quadrics. Moreover, we identify geometric conditions under which these maps are of bidegree (3,3). The examples of the latter include: 
\begin{itemize}
\item previously known Kahan discretizations of the Euler top and of the Zhukovski-Volterra gyrostat with one non-vanishing component of the gyrostatic momentum, 
\item a novel Kahan-type discretizations for the case of the Zhukovski-Volterra gyrostat with two non-vanishing components of the gyrostatic momentum, for which the usual Kahan discretization is non-integrable.
\end{itemize}
We expect that relaxing some of the restrictive geometric conditions will lead to an integrable Kahan-type discretization of general Nambu systems in $\mbR^3$ with quadratic Hamiltonians.

It can be anticipated that further research in this direction will lead to the discovery of a number of novel beautiful geometric constructions of integrable maps in dimension three and higher, related to addition laws on elliptic rational surfaces and on more complicated Abelian varieties. This will mark a further progress in the theory of integrable systems, under the general motto ``Geometry rules!''

\subparagraph*{Acknowledgement}
This research is supported by the DFG Collaborative Research Center TRR 109 ``Discretization in
Geometry and Dynamics''.

\subparagraph*{Figures} All figures have been made with \emph{Mathematica}.

\subparagraph*{Data availability} Data sharing is not applicable to this article as no data sets were generated or analyzed during the current study.

\subparagraph*{Conflicts of interest} The authors have no relevant financial or non-financial interests to disclose.

\bibliographystyle{acm}
\bibliography{Refs-2022}
\end{document}